\documentclass[conference]{IEEEtran}
\usepackage{amsmath,amsthm,amsfonts,braket,algorithm,graphicx}

\theoremstyle{definition} 
\theoremstyle{definition} 
\newtheorem {theorem} {Theorem}

\newcommand{\bk}[1]{\braket{#1|#1}}
\newcommand{\kb}[1]{\mathbf{\left[#1\right]}}

\newcommand{\M}{\texttt{Measure-Resend}}
\newcommand{\R}{\texttt{Reflect}}

\newcommand{\Hmin}{H_\infty}

\newcommand{\trd}[1]{\left|\left|{#1}\right|\right|}

\title{Semi-Quantum Random Number Generation}

\author{\IEEEauthorblockN{Julia Guskind}
\IEEEauthorblockA{\textit{Department of Computer Science and Engineering} \\
\textit{University of Connecticut}\\
Storrs, CT USA}
\and
\IEEEauthorblockN{Walter O. Krawec}
\IEEEauthorblockA{\textit{Department of Computer Science and Engineering} \\
\textit{University of Connecticut}\\
Storrs, CT USA \\
walter.krawec@uconn.edu}
}

\begin{document}
\maketitle

\begin{abstract}
Semi-quantum cryptography involves at least one user who is semi-quantum or ``classical'' in nature.  Such a user can only interact with the quantum channel in a very restricted way.  Many semi-quantum key distribution protocols have been developed, some with rigorous proofs of security.  Here we show for the first time, to our knowledge, that quantum random number generation is possible in the semi-quantum setting.  We also develop a rigorous proof of security, deriving a bound on the random bit generation rate of the protocol as a function of noise in the channel.  Our protocol and proof may be broadly applicable to other quantum and semi-quantum cryptographic scenarios where users are limited in their capabilities.
\end{abstract}

\section{Introduction}

Quantum Random Number Generators (QRNG) are an important cryptographic tool.  By utilizing the physical randomness of a quantum source, such protocols can distill cryptographically-secure random strings which are, themselves, important necessities for other cryptographic primitives (e.g., encryption, key distillation, and so on).  By now there are many QRNG protocols along with several security models.  Security models range from the extreme device independent scenario \cite{qrng-di-1,qrng-di-2,qrng-di-exp1,qrng-di-exp2} (where all measurement and source devices are untrusted) to the fully trusted scenario (where all devices are trusted and completely characterized).  An interesting middle-ground is the \emph{source independent} model whereby measurement devices are trusted, but source devices may be offloaded to an untrusted party \cite{vallone2014quantum,haw2015maximization,xu2019high}.  Such systems lead to a good security middle-ground while also providing fast experimental bit generation rates \cite{avesani2018source,drahi2020certified}.  For a general review of QRNG protocols the reader is referred to \cite{herrero2017quantum}; for a more general review on quantum cryptography, including theoretical and experimental developments, the reader is referred to \cite{QKD-survey1,QKD-survey2,amer2021introduction}.

Semi-quantum cryptography was first introduced in \cite{SQKD-first} for the key-distribution problem (QKD).  In general QKD protocols, like all other quantum cryptographic protocols, require parties to be ``fully-quantum'' or ``quantum-capable.'' Namely, they must be able to manipulate quantum bits in arbitrary ways.  Semi-quantum protocols involve at least one user who has restrictions and is almost classical in nature.  This restricted user is only able to interact with the quantum channel in a limited, classical, way: the user may measure and send in a single, publicly known basis (usually the $\ket{0},$ $\ket{1}$ basis) or to ignore the quantum channel, disconnecting from it, and returning all received signals back to the sender undisturbed.  Clearly, if all parties were restricted to these operations, the entire protocol would be mathematically equivalent to a classical one.  Thus, the restricted party is often called the ``classical user.''   The goal of any semi-quantum protocol is to achieve the same end-result as the original fully-quantum version (e.g., unconditionally secure key distribution) but using fewer resources.  In a way, they help us to study the ``gap'' between classical and quantum protocols and help us answer the question ``how quantum'' do protocols need to be to gain an advantage over its classical counterpart?

Most semi-quantum protocols are restricted to key-distribution \cite{SQKD-second,zou2009semiquantum,krawec2014restricted,amer2019semiquantum,zhang2017fault,krawec2015mediated,tsai2018semi,vlachou2018quantum,iqbal2020high}, however there are other primitives that are also available, in particular secret sharing \cite{secret-1,secret-2,secret-3}, secure direct communication \cite{SDC-1,SDC-2,SDC-3,SDC-4,SDC-5}, private comparison \cite{priv-comp-1,priv-comp-2}, and identity authentication \cite{ident-1,ident-2}.  Semi-quantum protocols can also be experimentally feasible \cite{boyer2017experimentally,massa2022experimental}.  For a general survey of semi-quantum cryptography, the reader is referred to \cite{SQKD-survey}.  To our knowledge, no semi-quantum random number generation (SQRNG) protocol is available.  We prove in this paper that QRNG is a viable semi-quantum primitive by designing the first SQRNG protocol and also deriving a rigorous information theoretic proof of security for it.

We make several contributions in this work.  We develop a new, and to our knowledge the first, semi-quantum random number generation protocol allowing a ``classical'' or semi-quantum user to generate a cryptographically-secure random number.  This is done using a fully-quantum server as a source and partial measurement device.  However, this server need not be trusted and, in fact, may be fully controlled by an adversary.  We develop an information theoretic proof of security for our protocol and derive an asymptotic random-bit generation rate, as a function of observed noise in the channel connecting the semi-quantum user to the adversarial server.  Our proof of security involves reducing the semi-quantum protocol to an entanglement-based version and deriving a bound on the quantum entropy of the resulting system.  The protocol and security proof methods may be important foundational building-blocks for future work in (semi) quantum cryptography.

\section{Preliminaries}

In this section we discuss some basic notation and concepts.  Let $x\in\{0,1\}^n$, then we define $ct_i(x)$ to be the number of times $i\in\{0,1\}$ appears in the string $x$.  Given a quantum state $\ket{\psi}$ we write $\kb{\psi}$ to mean $\ket{\psi}\bra{\psi}$.  Given a state $\rho_{AE}$ acting on some Hilbert space $\mathcal{H}_A\otimes\mathcal{H}_E$, we write $\rho_A$ to mean the state resulting from the partial trace over the $E$ system, namely $\rho_{A} = tr_E\rho_{AE}$.  Similarly for multiple systems.

The \emph{von Neumann entropy} of a state $\rho_A$ is defined to be $S(A)_\rho = -tr(\rho_A\log\rho_A)$ where all logarithms in this paper are base two unless otherwise specified.  Given state $\rho_{AE}$ we write $S(A|E)_\rho$ to mean the conditional von Neumann entropy defined to be $S(A|E)_\rho = S(AE)_\rho - S(E)_\rho$.

We will be working with quantum random number generation (QRNG).  For this, a quantum protocol is run whereby Alice and Eve holds two registers - Alice's register is a classical system containing her \emph{raw random string} while Eve's system is quantum and arbitrarily entangled.  This raw random string may not be uniform or completely independent of Eve, thus a further post-processing step is required.  Given a classical-quantum state $\rho_{AE}$, we are interested in how much uniform randomness, independent of $E$, can be extracted from $A$.  The process of \emph{privacy amplification} can be used to extract this randomness.  Namely, if one chooses a random two-universal hash function sending $N$-bit strings to $\ell$-bit strings, and maps the $A$ register through this hash function, resulting in $\rho_{F(A)E}$, then it was shown in \cite{renner2008security} that:
\begin{equation}
  \trd{\rho_{F(A)E} - \frac{I}{2^\ell}\otimes \rho_E} \le 2^{-\frac{1}{2}(\Hmin^\epsilon(A|E)_\rho - \ell)} + 2\epsilon,
\end{equation}
where $\Hmin^\epsilon(A|E)$ is the smooth quantum min entropy \cite{renner2008security}.  Thus, to ensure that the resulting system is $\epsilon'$-close to a uniform random string, independent of any adversary system, one requires a bound on the quantum min entropy.  In the asymptotic scenario (which we consider in this paper), where $|A|\rightarrow \infty$ and $\epsilon \rightarrow 0$, and assuming collective attacks whereby $\rho_{AE}$ takes a product form (and so $\rho_{AE} = \sigma_{AE}^{\otimes N}$), it holds that $\frac{1}{N}\Hmin^\epsilon(A|E)_{\rho} \rightarrow S(A|E)_\sigma$ (see \cite{tomamichel2009fully}).  From this, it is clear that:
\begin{equation}\label{eq:entropy-bit-rate}
\lim_{\epsilon\rightarrow 0}\lim_{N\rightarrow \infty}\frac{\ell}{N} = S(A|E)
\end{equation}
Note the similarities to the Devetak-Winter key-rate for QKD \cite{QKD-Winter-Keyrate}; the only difference is that there is no need to worry about error correction leakage.

\section{The Protocol}
We now introduce our SQRNG protocol.  The protocol consists of a semi-quantum user Alice and a fully quantum server capable of creating quantum signals and measuring them in the $X$ basis (see Figure \ref{fig:protocol}).  A semi-quantum user is restricted to performing operations in a single, publicly known, basis (usually the computational basis $\{\ket{0}, \ket{1}\}$) or to ignoring the quantum signal and reflecting it undisturbed.  Formally, the semi-quantum user is restricted to performing one of the following two operations on every qubit received from some sender:
\begin{enumerate}
\item $\M$: the semi-quantum user subjects the qubit to a $Z$ basis measurement resulting in outcome $r \in \{0,1\}$.  A qubit in the state $\ket{r}$ is then returned to the sender.
\item $\R$: the semi-quantum user ignores the qubit and reflects it, undisturbed, back to the sender.
\end{enumerate}

\begin{figure}
    \centering
    \includegraphics[width=.35\textwidth]{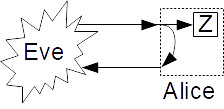}
    \caption{The proposed SQRNG protocol.  An untrusted quantum server (Eve) sends quantum signals to semi-quantum or ``classical'' Alice.  If honest, the server should send $\ket{+}$.  Alice may only measure in the computational $Z$ basis or reflect the incoming signal.  The server, if honest, will measure the returning signal in the $X$ basis and report the outcome.  We prove security assuming the server is adversarial and does not necessarily prepare the correct state or perform the correct measurement.}
    \label{fig:protocol}
\end{figure}

Since the quantum server is untrusted, we call the server Eve.  Though we will later, in our security analysis, assume the server is adversarial, for the purposes of this section, we explain the protocol's operation assuming an honest server.  The protocol operates over several rounds, a single round consisting of the following actions:
\begin{enumerate}
    \item The server prepares the state $\ket{+}$ and sends it to Alice.
    \item Alice chooses randomly to $\R$ the signal or to $\M$, (recording the measurement result).
    \item Upon receiving Alice's qubit, the server makes an $X$ basis measurement and reports the outcome over a classical channel (this channel need not be authenticated).
      \begin{itemize}
      \item If Alice chose $\R$, she should expect the message from the server to be ``$+$'' and any other response will be considered noise and will be later factored into our security analysis.
        \item If Alice chose $\M$, she will have a classical measurement outcome $a \in \{0,1\}$.  Ideally, the server's message will be either ``$+$'' or ``$-$'' randomly in this event.
      \end{itemize}
\end{enumerate}

The above quantum portion of the protocol is repeated $N$ times leading to a raw random string of size $n \le N$.  This raw random string is the result of Alice's measurements when she chose $\M$.  However, it may not be uniform (due to adversarial or natural noise) and Eve may have some side information on it through her quantum ancilla and her operation as the server (i.e., she need not perform an honest $X$ basis measurement on step 3).  Thus, it is next processed through privacy amplification, hashing the raw string using a two-universal hash function.  As proven in \cite{frauchiger2013true}, for QRNG, this hash function need only be chosen once and thus additional randomness is not needed on Alice's part to choose this hash function.

It is worth discussing how much seed randomness Alice needs to execute the above protocol.  She requires some secret random string in order to make her operational choices each round.  We may consider the protocol differently, in that Alice will choose a subset $t$ of size $m$ which indexes those rounds (out of the $N$ total rounds) that will be used for testing, namely those rounds where she chooses $\R$.  For all rounds not in $t$ she will choose $\M$.  Choosing this subset requires $\log{N\choose m}$ bits of seed randomness.  This seed may be replenished by the random string produced after privacy amplification.  However, as we are considering the asymptotic scenario in this paper, we may choose $m$ to be arbitrarily small compared to $N$ and, so, as $N\rightarrow \infty$, it will hold that $\frac{1}{N}\log{N\choose m} \rightarrow 0$ and so we will not consider this in our subsequent analysis.  In the finite signal scenario, of course, it would be important.



\section{Security Analysis}
We now prove security of our protocol against general attacks in the asymptotic scenario.  The security model used by our proof assumes the following:
\begin{enumerate}
\item The server is adversarial and may produce any arbitrary quantum state in Step 1 of the protocol.  This state may be entangled with the server's private ancilla.  We assume the dimension of the quantum state actually sent to Alice is $2^N$ where $N$ is the number of rounds used by the protocol (a parameter set by the user Alice).  That is, the server is allowed to perform any general attack in Step 1, however we assume each round consists of an ideal qubit.
\item The channel connecting the server to Alice is noisy but not lossy.  However, the server, may replace the channel with a perfect one (as is typically assumed in quantum cryptographic proofs) and thus any noise detected by Alice is the result of an adversary.  This assumption is to the benefit of the adversary as any natural channel noise can only cause more uncertainty for Eve.

\item Alice's measurement devices are ideal and qubits are ideal.  Thus, we are not considering practical attacks such as photon tagging attacks \cite{SQKD-photon-tag,SQKD-photon-tag-comment}

\item On Step 3 of the protocol, the server may perform any quantum operation on the returning $N$ qubit signal and her ancilla from Step 1.  However, the server must send a classical message of size $N$-bits (a one-bit message for every round of the protocol).  This attack is modeled as a \emph{quantum instrument} \cite{davies1970operational}.

  \item While there may be third party adversaries, we assume that they collude with the server.  Thus, any third-party adversary attack may be ``absorbed'' into the adversarial server's attack strategy and we need only prove security against adversarial servers.  Note that this includes attacks by third-party adversaries against the unauthenticated classical channel connecting the server to Alice.  This assumption is to the benefit of the adversary.
\end{enumerate}


Note, though we are proving security against general attacks here, we are not considering practical attacks.  While practical attacks are important, proving security against general attacks even in the ideal-qubit scenario is challenging in itself and usually a first step (especially for semi-quantum protocols).  Mirror-style adaptions may be useful here \cite{boyer2017experimentally,krawec2018practical}; however, we consider practical attacks as interesting future work, but out of scope for this paper.

Our proof works in two stages.  First, we show how our protocol can be reduced to an equivalent entanglement based version.  For this we will adapt methods similar to those we developed in \cite{guskind2022mediated} for a particular semi-quantum key distribution protocol.  This entanglement protocol is actually a QRNG protocol (i.e., a fully-quantum protocol), where the user is now fully quantum.  Nonetheless, we show that security of that QRNG protocol implies security of the SQRNG version.  The second step proves security of this entanglement based version (thus proving security of the actual semi-quantum prepare-and-measure protocol we developed).  Note that, while we only consider this particular SQRNG protocol, we suspect our methods and our proof methodology may be broadly applicable to other quantum cryptographic protocols, especially those with device limitations and two-way quantum communication, which are generally a challenge to prove using standard proof techniques.

\subsection{Reduction to an Entanglement Based Protocol}
We first show how our SQRNG protocol may be reduced to an entanglement based version denoted by e-QRNG.  Our reduction is based on methods we developed in \cite{guskind2022mediated} for \emph{mediated SQKD} protocols (key-distribution protocols using a quantum server \cite{krawec2015mediated}); we show here that these methods can be extended to work with our SQRNG protocol.

We first comment that the SQRNG protocol may be ``purified'' in the following sense.  Instead of Eve preparing each qubit individually, she prepares a large $N$-qubit state where $N$ is the number of rounds used by the protocol, possibly entangled with her private ancilla.  There is no assumption that this state be a product state.  Next, Alice, on receipt of these qubits, will choose a subset $\Theta$.  For each $i=1, \cdots, N$, if $\Theta_i = 0$, she will $\R$ the $i$'th signal (namely perform an identity operation); otherwise she will $\M$ the qubit.  However, instead of actually performing a destructive measurement at this point, Alice will simply apply a CNOT gate with the control register being the $i$'th qubit sent from the server and the target being a blank qubit cleared to the $\ket{0}$ state.  She may then later measure this private register in the $Z$ basis.  It is clear that this later measurement will simulate the case where Alice measures the qubits immediately.  Following this operation by Alice (either the Identity operation or the CNOT operation), the entire signal of $N$ qubits returns to Eve who performs a quantum instrument mapping the $N$ qubits and her initial private ancilla to a classical $N$-bit message space (modeled as a quantum ancilla) and an updated ancilla from which Eve will attempt to learn as much as possible about Alice's measurement results (or, rather, Alice's private register which she will subsequently measure).  Quantum instruments may be used to model scenarios where a quantum state undergoes some quantum operation (including potential measurements) yielding a classical and quantum state as output. The classical part represents the message Eve sends to Alice while the quantum part represents Eve's ancilla in the event she chose to send that particular message.  In the purified case, this instrument may be dilated to an isometry (and subsequently a unitary) operator using standard techniques \cite{wilde2011classical}.  The fact that this purified version is identical to the standard prepare-and-measure one follows through standard methods; the advantage to it is that the quantum state remains a pure state which will be easier to argue about in terms of security.

Now, we show how this purified SQRNG protocol can be reduced to an entanglement QRNG, denoted e-QRNG.  The e-QRNG protocol is a fully quantum protocol (not a semi-quantum one) with three parties: Alice, a Trusted Server (who may also be Alice), and Eve.  Alice will have choices similar to the SQRNG case, namely $\M$ and $\R$ though these operations are different here and, in fact, have no direct operational meaning in the e-QRNG protocol.  We show later that, the operation of $\R$ in the e-QRNG case (which is not actually a reflection but a particular measurement) will ``simulate'' a true $\R$ operation in the SQRNG case; similarly for the $\M$ operation.

The entanglement protocol (e-QRNG) operates as follows:
\begin{enumerate}
    \item A quantum source, Eve, prepares an entangled state $\ket{\tau}_{ACE} \in \mathcal{H}_A \otimes \mathcal{H}_C \otimes \mathcal{H}_E$.  The $A$ and $C$ registers are each of dimension $2^N$ for user specified $N$ (as before, $N$ is the number of rounds used by the protocol).  The $A$ system is sent to Alice while the $C$ register is sent to the trusted server (who may also be Alice).  Eve keeps the $E$ register private.
    \item Alice, on receipt of the state, has two choices for each round (i.e., each of the $N$ qubits), as in the SQRNG protocol.  We call these two choices here $\M$ and $\R$ to keep the notation consistent with the SQRNG case, however the operations are different in the e-QRNG case.  If Alice chooses $\R$, she will measure that qubit in the $X$ basis, and \emph{abort the entire protocol} if she observes $\ket{-}$. That is, to simulate a true ``reflection'' in the semi-quantum case, $A$ will only continue with the protocol if she measures and observes $\ket{+}$.  We say Alice \emph{accepts} the state if she observes $\ket{+}$. Alternatively, if Alice chooses $\M$, she will measure that particular qubit in the $Z$ basis.  
    \item The trusted server, on receipt of the $C$ system, simply measures all $N$ qubits in the $X$ basis, and \emph{publicly reports the outcome}.  Note that in the e-QRNG case, the server is honest in this measurement and, in fact, this $C$ register measurement may even be done by Alice.  Note that this may be done before, in parallel to, or after Alice's operations in step 2.
\end{enumerate}

We comment that this e-QRNG protocol is highly inefficient due to the high probability of aborting (regardless of the noise level).  \emph{However, this e-QRNG protocol is not one users would actually run, but rather only used as a theoretical tool to prove security of the actual SQRNG protocol (which only aborts if the noise level is too high).}  It is important to note that the protocol completely aborts if Alice observes $\ket{-}$ on her $\R$ test case.  The reason for this is that, as we show later, an observation of $\ket{+}$ will exactly simulate a true $\R$ in the SQRNG case whereas an observation of $\ket{-}$ produces a quantum system that cannot exist in the SQRNG case and, so, should not be analyzed.

\begin{theorem}\label{thm:reduction}
  Let $\mathcal{E}$ be an attack against the SQRNG protocol (which, itself, consists of an initial state description and a quantum instrument to apply after Alice's operation) and let $\Theta$ be Alice's choice of operation in the SQRNG protocol.  Let $\ket{\psi(\mathcal{E},\Theta)}_{AME}$ be the resulting state of operating the purified SQRNG protocol under these conditions (where $A$ is Alice's register; $E$ is Eve's private ancilla; and $M$ is the register storing the classical message sent from the adversarial server to Alice).  Then, there exists a quantum state $\ket{\tau(\mathcal{E})}_{ACE}$ for the e-QRNG protocol, depending only on $\mathcal{E}$, such that: (1) the probability of Alice not aborting, assuming she measures those qubits in $A$ with $\Theta_i = 0$, (i.e., she accepts $\ket{\tau}$) is non-zero and exactly $p_a = 1/2^{ct_0(\Theta)} > 0$ for any $\Theta$;  (2) conditioning on accepting, let $\ket{\tau(\mathcal{E}, \Theta)}$ be the resulting quantum state (which clearly does depend on $\Theta$), then it holds that $\ket{\tau(\mathcal{E},\Theta)} = \ket{\psi(\mathcal{E}, \Theta)}$; and (3) if the SQRNG protocol is attacked using a collective attack (namely the produced state is a product state), then the constructed state $\ket{\tau}$ is also a product state $\ket{\tau} = \ket{\tau_0}^{\otimes N}$ (the result of a collective attack) and, furthermore, the probability of observing a $\ket{-}$ on any individual signal in the e-QRNG case is exactly $1/2$.
\end{theorem}
\begin{proof}
Let's consider the (purified) SQRNG protocol.  In this case, a general attack consists of Eve first preparing an arbitrary quantum state $\ket{\psi_0}_{AE}$ which, without loss of generality, we may write as:
\[
\ket{\psi_0} = \sum_{a \in \{0,1\}^N}\alpha_a\ket{a}\ket{E_a},
\]
where the $\ket{E_a}$ are arbitrary normalized states in Eve's private ancilla.  The $A$ portion is sent to Alice who chooses $\Theta$ which dictates whether to $\M$ or to $\R$.  Whenever $\Theta_i = 1$ she will apply a CNOT gate as discussed.  Since her ancilla is initially cleared to a zero state, the result of this action on $\ket{\psi_0}$ is easily seen to be:
\[
\ket{0\cdots 0}_A\otimes\ket{\psi_0}\rightarrow \ket{\psi_0'} = \sum_{a\in\{0,1\}^N}\alpha_a\ket{a\wedge\Theta,a}_{AT}\ket{E_a},
\]
where $a\wedge\Theta$ denotes the bit-wise AND operation and the $\ket{a\wedge\Theta}_A$ is Alice's private ancilla after applying CNOT gates whenever $\Theta_i = 1$.  The $T$ (Transit) register represents the quantum state that will return to Eve.

At this point, the $T$ register will return to Eve's control who applies a quantum instrument to the system.  As shown in \cite{krawec2015mediated}, this is equivalent to applying an isometry $U$ mapping Eve's ancilla and the Transit register into a quantum ancilla for Eve and a Hilbert space spanned by all possible classical messages that could have been sent (namely, a Hilbert space of dimension $2^N$ since for every round there are only two classical messages Eve is allowed to send).  After applying the isometry $U$ to the returning system and Eve's ancilla from her initial state preparation, she measures the message Hilbert space. The measurement outcome determines the message transcript she sends to Alice while the post measured system represents her quantum ancilla in the event she had sent that message using the quantum instrument attack.

Without loss of generality, we may define $U$'s action on basis states as follows:
\[
U\ket{a, E_a} = \sum_{m\in\{+,-\}^N}\ket{m, F_{a,m}}_{M,E},
\]
where the $M$ register is Eve's classical message.  Note that $U$'s action on basis states of the form $\ket{a, E_b}$ for $a \ne b$ may be arbitrary as these do not appear in the system under investigation.

Applying $U$ to the returning state $\ket{\psi_0'}$ yields:
\begin{equation}\label{eq:real-state}
  \ket{\psi} = \sum_{a\in\{0,1\}^N}\alpha_a\ket{a\wedge\Theta}\otimes\sum_{m\in\{+,-\}^N}\ket{m, F_{a,m}}_{cl,E}
\end{equation}
Eve will then measure the $M$ register dictating her message.

We now claim that there is an equivalent attack strategy against the e-QRNG protocol, producing an identical quantum system in the event the entanglement based protocol does not abort.  That is, any attack against the SQRNG protocol can be translated to an equivalent attack against the e-QRNG protocol and, therefore, if the e-QRNG protocol is secure, the SQRNG protocol must also be (since e-QRNG can only have more attacks against it).  

Consider, now, the entanglement based protocol.  Assume Eve prepares the initial state:
\begin{align}
  \ket{\tau}_{ACE} &= U\left(\sum_{a\in\{0,1\}^N}\alpha_a\ket{a, E_a}\right)\notag\\
  &= \sum_{a\in\{0,1\}^N}\alpha_a\ket{a}_A \otimes \sum_{m\in\{+,-\}^N}\ket{m, F_{a,m}}_{CE}.\label{eq:entangled-state}
\end{align}
The $A$ and $C$ registers are sent to Alice and the trusted server respectively (as discussed, the trusted server may in fact be Alice in the e-QRNG protocol case).  

We first show the second claim of the theorem namely that, conditioning on Alice accepting the e-QRNG state, it agrees with the SQRNG state.  For this, we note that we can decompose the sum over $a\in\{0,1\}^N$ in $\ket{\psi}$ (Equation \ref{eq:real-state}) into two parts: a part where $\Theta_i = 0$ and a part where $\Theta_i = 1$.  In particular, we can write $a = \pi_\Theta(a_0, a_1)$ for some permutation $\pi$ depending on $\Theta$.  Here, $|a_0| = ct_0(\Theta)$ and $|a_1| = ct_1(\Theta)$.  The function $\pi_\Theta$ simply maps the first argument into the appropriate bit-position of $a$ where $\Theta$ is zero; similarly for the second argument.  Thus Equation \ref{eq:real-state} (the SQRNG case) becomes:
\begin{equation}
  \ket{\psi} = \sum_{a_0, a_1}\alpha_{\pi_\Theta(a_0,a_1)}\ket{\pi_\Theta(0\cdots 0,a_1)}\otimes\sum_m\ket{m, F_{\pi_\Theta(a_0,a_1),m}}.
\end{equation}
Of course the sum above is over all $a_0\in\{0,1\}^{ct_0(\Theta)}$ and $a_1\in\{0,1\}^{ct_1(\Theta)}$.

Now, we may similarly decompose the initial state prepared by Eve in the e-QRNG case (Equation \ref{eq:entangled-state}) as follows:
\begin{align*}
  \ket{\tau} &= \sum_{a_0,a_1} \alpha_{\pi_\Theta(a_0,a_1)}\ket{\pi_\Theta(a_0,a_1)}\otimes\sum_m\ket{m,F_{\pi_\Theta(a_0,a_1)}}\\\\
  &=\frac{1}{\sqrt{2^{m}}}\sum_{a_0,a_1}\alpha_{\pi_\Theta(a_0,a_1)}\ket{\pi_\Theta(+^m, a_1)} \sum_m\ket{m, F_{\pi_\Theta(a_0,a_1)}}\\
  &+ \ket{\mu}_{AE}\\\\
  &\cong \frac{1}{\sqrt{2^{ct_0(\Theta)}}}\ket{\psi} + \ket{\mu}.
\end{align*}
where $m = ct_0(\Theta)$, $+^m = +\cdots +$ ($m$ times), and $\ket{\mu}_{AE}$ is some state that has at least one $\ket{-}$ where $\Theta_i = 0$ in Alice's register (and, thus, would lead to the protocol aborting).  It is clear, then, that conditioned on Alice not aborting the e-QRNG protocol (i.e., accepting the state $\ket{\tau}$), the state collapses to $\ket{\psi}$, the same state that would have been produced if the SQRNG protocol had been run.

It is also clear that, since $\bk{\psi} = 1$ and $\braket{\psi|\mu} = 0$, the probability of Alice not aborting is strictly positive (and, thus, we do not condition on a probability zero event).  In fact, the probability of not aborting (i.e., accepting) is exactly $\frac{1}{\sqrt{2^{ct_0(\Theta)}}}$ proving claim (1) of the theorem.

Finally, to prove claim (3), we note that if Equation \ref{eq:real-state} were produced by a collective attack, then Equation \ref{eq:entangled-state} would be a product state also and the probability of accepting will remain $\frac{1}{\sqrt{2^{ct_0(\Theta)}}}$ implying that the probability of accepting any particular signal is exactly $1/2$ completing the proof.
\end{proof}

Theorem \ref{thm:reduction} implies that any attack against the SQRNG protocol (whose security we want to prove) can be translated to an attack against the e-QRNG protocol which: (1) produces the same quantum system for Alice and Eve conditioned on Alice accepting the e-QRNG state (thus any entropy computation will be identical and any observed statistics will also be identical in the accepting case); and (2) the probability of accepting is strictly positive and known.  Note that, even though the e-QRNG protocol is highly inefficient, this is not relevant as we are only interested in bounding the quantum entropy of the e-QRNG protocol conditioned on a non-abort.  This will translate directly to a bound on the entropy of the SQRNG protocol (which never aborts, unless Alice determines the noise is too high - a threshold which we can compute later).  Thus, even though the e-QRNG protocol is highly inefficient, this does not matter as it is only a theoretical tool for the security proof and not an actual protocol to run in practice.  Note, also, that there are many more attacks against the e-QRNG protocol, including attacks which would cause it to always abort; however analyzing those ``denial of service'' attacks are not relevant as they would never appear in the SQRNG protocol.


\subsection{Secure Bit Rate Analysis}
Our goal is to derive an asymptotic bit generation rate for the SQRNG protocol.  Consider a run of the SQRNG protocol where Eve employed some (unknown) attack described by $\mathcal{E}$, and $\Theta$ was Alice's choice of operations resulting in state $\ket{\psi_{SQRNG}}$.  From Theorem \ref{thm:reduction}, there exists an equivalent quantum state $\ket{\psi}$ produced by the e-QRNG protocol.  We will derive a bit generation rate for this e-QRNG state which will translate directly to a bit-generation rate for the SQRNG protocol.  Since $\mathcal{E}$ and $\Theta$ were arbitrary, our method will work for any attack and choice of $\Theta$ for the SQRNG protocol thus proving the SQRNG protocol secure.


We first assume collective attacks for the SQRNG protocol (and thus, by condition (3) of Theorem \ref{thm:reduction} also for the e-QRNG protocol state); namely, the state $\ket{\psi}$ may be described as a product state $\ket{\psi} = \ket{\mu}^{\otimes N}$ with the probability of accepting any particular signal state $\ket{\mu}$ is $1/2$  (i.e., the probability of Alice observing a $\ket{+}$ in $\ket{\mu}$ is $1/2$) .

Now, let's consider an individual signal state $\ket{\mu}$.  We may write this in the most general way as follows:
\begin{equation}\label{eq:collective-state}
\ket{\mu} = \sum_{a,c\in\{0,1\}}\ket{a,c}_{AC}\otimes\ket{e_{a,c}}_E,
\end{equation}
where the $\ket{e_{a,c}}$ states are arbitrary (not necessarily normalized nor orthogonal) states in Eve's ancilla.  Note that, when $c = 0$ in the summation we actually mean a state of $\ket{+}$ while $c=1$ implies $\ket{-}$.  Our goal now is to compute a bound on $S(A|E)_\mu$ since this will give us our bit-rate according to Equation \ref{eq:entropy-bit-rate}. Our bound, of course, must be a function only of those statistics which may be observed by the users in the event the protocol does not abort.

First, consider $P_{a,c}^{AC} = Pr(A=a \wedge C=c)$.  From Equation \ref{eq:collective-state}, this is easily seen to be $P_{a,c}^{AC} = \bk{e_{a,c}}$.  Ideally, this should be $1/4$, though we do not assume anything about these values in our proof only that they may be observed.  Next, consider $P_{+|acc} = Pr(C =``+'' | \text{accept})$.  Changing basis, we have:
\[
\ket{\mu} = \frac{1}{\sqrt{2}}\ket{+}_A\otimes\left(\sum_c\ket{c}(\ket{e_{0,c}} + \ket{e_{1,c}})\right) + \frac{1}{\sqrt{2}}\ket{-}(\cdots)
\]
Since by Theorem \ref{thm:reduction} we know the probability of accepting (i.e., measuring a $\ket{+}$) is $1/2$ for each signal state independently, the conditional state collapses to:
\[
\ket{+}_A\otimes\left[\ket{0}_C(\ket{e_{0,0}} + \ket{e_{1,0}}) + \ket{1}_C(\ket{e_{0,1}} + \ket{e_{1,1}})\right],
\]
from which it is clear that $P_{+|acc} = \bk{e_{0,0}} + \bk{e_{1,0}} + 2Re\braket{e_{0,0}|e_{1,0}}$.  Similarly, we may define $P_{-|acc}$ and find it to be $P_{-|acc} = \bk{e_{0,1}} + \bk{e_{1,1}} + 2Re\braket{e_{0,1}|e_{1,1}}$.  Note that these probabilities coincide directly with the probability the server sends the message ``$+$''/``$-$'' conditioned on Alice choosing $\R$ in the SQRNG case.

We are now in a position to compute the bit generation rate.  From Equation \ref{eq:collective-state}, the system, conditioning on Alice actually distilling a random bit (namely she measures the $A$ register of $\ket{\mu}$ in the $Z$ basis) is easily found to be:
\begin{equation}
\rho_{ACE} = \kb{0}_A\otimes\left(\sum_c\kb{c, e_{0,c}}\right) + \kb{1}_A\otimes\left(\sum_c\kb{c,e_{1,c}}\right)
\end{equation}
Of course, Eve has access to the $C$ and and $E$ registers (since the trusted server makes its measurement results public); thus to compute the bit generation rate, we need to bound $S(A|CE)$.  Using Theorem 1 from \cite{QKD-Tom-Krawec-Arbitrary}, along with our analysis above of Eve's inner-products, we have the following lower-bound:
\begin{align}
S(A|CE) &\ge \left(P_{0,0}^{AC} + P_{1,0}^{AC}\right)\cdot\left(h\left[\frac{P_{0,0}^{AC}}{P_{0,0}^{AC}+P_{1,0}^{AC}}\right] - h\left[\lambda_0\right]\right)\label{eq:entropy-bound}\\
&+\left(P_{0,1}^{AC} + P_{1,1}^{AC}\right)\cdot\left(h\left[\frac{P_{0,1}^{AC}}{P_{0,1}^{AC}+P_{1,1}^{AC}}\right] - h\left[\lambda_1\right]\right),\notag
\end{align}
where:
\begin{equation}
\lambda_c = \frac{1}{2}\left(1 + \frac{\sqrt{(P_{0,c}^{AC} - P_{1,c}^{AC})^2 + 4Re^2\braket{e_{0,c}|e_{1,c}}}}{P_{0,c}^{AC} + P_{0,c}^{AC}}\right).
\end{equation}
The inner products needed to compute $\lambda_c$ may be determined from $P_{c|acc}$.  For instance, $2Re\braket{e_{0,0}|e_{1,0}} = P_{+|acc} - P_{0,0}^{AC} - P_{1,0}^{AC}$.  This gives us everything we need to compute the quantum entropy of the e-QRNG state conditioned on the protocol not aborting.  Since the state for the entanglement based protocol conditioned on not aborting is identical to that of the SQRNG protocol of the given, but arbitrary, attack and, furthermore, since all observable statistics in that case are also identical, the computed bit generation rate applies to the SQRNG protocol as desired.


The above rate applies to collective attacks against the SQRNG protocol in the asymptotic scenario.  However, for general attacks, the constructed e-QRNG state may be made permutation invariant in the usual way \cite{renner2008security} and then, de Finetti style arguments \cite{konig2005finetti} may be used to promote the above analysis to general attacks.  

\section{Evaluation}

To evaluate the bit generation rate of our SQRNG protocol, one only needs to observe those probability values appearing in the entropy expression Equation \ref{eq:entropy-bound}.  For the purposes of this paper, we simulate observable probability values assuming the source noise is modeled by a depolarization channel acting independently on each qubit and that the server is honest.  Note that this assumption is only used in this section to determine values for those probability values appearing in our bit generation expression above.  That is, this assumption is \emph{only used here for evaluation purposes and is not a required assumption in our security proof above.}  Normally, one would simply observe the actual probability values; however since we are performing a theoretical analysis and not an experiment, we must simulate ``reasonable'' values for these probabilities.  Depolarization channels are the most common ones evaluated theoretically.

Such a channel takes a qubit quantum state $\rho$ and maps it to $\mathcal{D_Q}:\rho \mapsto (1-2Q)\rho + Q\cdot I$, where $I$ is the identity operator (an equally mixed state of $\ket{+}$ and $\ket{-}$).  Of course, since we have a two-way channel, we will need to worry about the noise and behavior of the channel in both directions.  For this, we consider two scenarios: \emph{independent} channels and \emph{dependent} ones.  The distinction only matters when Alice chooses to $\R$ since, then, it is possible that noise in the reverse channel depends on noise in the forward one.  When Alice chooses $\M$, the qubit is measured so any channel dependence is broken.  More formally, we parameterize the noise in the forward and reverse channels by $Q$.  But in the event Alice chooses $\R$, the joint forward/reverse channel is modeled as a depolarization channel with parameter $Q_{FR}$.  For independent channels, the noise in the reverse channel acts independently of the forward for reflection events and, so, we set $Q_{FR} = 2Q(1-Q)$.  In the case of dependent channels, the noise in the reverse channel can depend on the forward and so we set $Q_{FR} = Q$.  Note that this behavior often appears in fiber implementations where any phase noise picked up in the forward direction is ``undone'' in the reverse channel assuming the photon is reflected back \cite{lucamarini2014quantum,beaudry2013security}.

From this parameterization, we may determine values for the needed probability values.  In particular, we will assume the state arriving at Alice's lab is of the form $\mathcal{D}_Q(\kb{+})$.  From this, it is clear that the probability of Alice measuring either a $\ket{0}$ or $\ket{1}$ is $1/2$.  Conditioned on such a measurement, a qubit returns to the server; the state, then, arriving at the server will be $\mathcal{D}_Q(\kb{a})$, where $a\in\{0,1\}$ is Alice's measurement outcome.  Since we are simulating an honest server in this section and a noisy channel, the probability that the server sends the message ``$+$'' is also $1/2$.  Thus, we find $P_{a,c}^{AC} = 1/4$ for all $a,c$.

Next, consider the case of a reflection.  In this case, the state returning to the server is of the form $\mathcal{D}_{Q_{FR}}(\kb{+})$.  From this, we see that $P_{+|acc} = (1-Q_{FR})$.  Note that $P_{+|acc}$ is technically defined only for the e-QRNG protocol, however from Theorem \ref{thm:reduction} it exactly coincides with the probability that the (potentially adversarial) server sends the message ``$+$'' conditioned on Alice choosing $\R$ in the SQRNG case.  This gives us everything we need to compute Equation \ref{eq:entropy-bound}.  In fact, under these conditions, our entropy bound simplifies significantly to:
\begin{equation}
S(A|CE) \ge 1 - h\left(1-Q_{FR}\right).
\end{equation}
This can be seen by noting that all $P_{a,c}^{AC}$ are $1/4$ and that $\lambda_0 = \lambda_1 = 1-Q_{FR}$ since $4Re^2\braket{e_{0,0}|e_{1,0}} = (1-Q_{FR} - 1/2)^2$ and $4Re^2\braket{e_{0,1}|e_{1,1}} = (-Q_{FR} + 1/2)^2$.  A graph of the resulting bit generation rate is shown in Figure \ref{fig:bit-rate}.  Rather interestingly, in the event of a dependent channel, the SQRNG protocol matches the bit-generation rate of the fully-quantum QRNG protocol introduced in \cite{vallone2014quantum}, at least in the asymptotic ideal qubit scenario. 

\begin{figure}
    \centering
    \includegraphics[width=.45\textwidth]{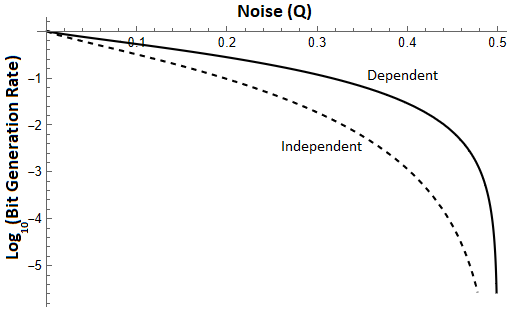}
    \caption{Showing the bit-generation rate of our SQRNG protocol under depolarization noise.  Both the dependent and independent channel cases are plotted.}
    \label{fig:bit-rate}
\end{figure}

\section{Closing Remarks}

We developed a novel, and to our knowledge the first, semi-quantum random number generation protocol.  We also derived a rigorous information theoretic proof of security for the protocol, by reducing it to an entanglement-based protocol and deriving a bound on its random bit generation rate.  Our evaluations showed that the bit-generation rate can match that of other fully-quantum QRNG protocols, at least in the asymptotic ideal-qubit scenario.  Combined with recent research in semi-quantum key distribution, which shows that the asymptotic behavior of SQKD protocols can match BB84 \cite{SQKD-survey}, our work in this paper shows similarly optimistic results for the QRNG case.

Many interesting future problems remain open.  Of great importance would be to study SQRNG protocols under more practical considerations.  Here certain practical device attacks are problematic such as the photon-tagging attack \cite{SQKD-photon-tag,SQKD-photon-tag-comment}.  However, mirror style devices may mitigate this \cite{boyer2017experimentally}.  Adapting mirror-style devices to the SQRNG case would be an interesting problem.  Also of importance would be deriving a bound on the random bit generation rate in the finite signal setting (as opposed to the asymptotic case considered here).  Our reduction Theorem \ref{thm:reduction} applies here, however when analyzing the finite key scenario, one requires bounds on the quantum min entropy which is harder to derive.  Thus, we leave this as an open problem, though our reduction in Theorem \ref{thm:reduction} should prove a valuable tool in that investigation.

\section*{Acknowledgment}
WOK would like to acknowledge support from NSF grant number 2143644.


\end{document}